\documentclass[10pt,a4paper]{article}

\emergencystretch=\maxdimen
\usepackage{anysize}
\usepackage{listings}
\usepackage{graphicx}
\usepackage{float}
\usepackage{geometry} 
\usepackage{booktabs}
\usepackage{array} 
\usepackage{paralist} 
\usepackage{verbatim} 
\usepackage{subfig} 
\usepackage{multirow}
\usepackage{textcomp}
\usepackage[hidelinks, breaklinks]{hyperref}
\usepackage{color}
\usepackage{eurosym}
\usepackage{threeparttable}
\usepackage{rotating}
\usepackage{lipsum}
\usepackage{amsmath}
\usepackage{amssymb}
\usepackage{amsfonts}
\usepackage{xcolor}
\usepackage{xfrac}
\usepackage{calc}
\usepackage{tabularx}
\usepackage{lipsum}
\usepackage{setspace}
\usepackage{sectsty}
\usepackage{placeins}
\usepackage{array, multirow}
\usepackage{enumitem}  
\usepackage{calc}
\usepackage{natbib}
\usepackage{caption}
\usepackage{titlesec}
\titlelabel{\thetitle.\quad}
\newcolumntype{L}[1]{>{\raggedright\let\newline\\\arraybackslash\hspace{0pt}}m{#1}}
\newcolumntype{C}[1]{>{\centering\let\newline\\\arraybackslash\hspace{0pt}}m{#1}}
\newcolumntype{R}[1]{>{\raggedleft\let\newline\\\arraybackslash\hspace{0pt}}m{#1}}
\newtheorem{theorem}{Theorem}

\newtheorem{definition}[theorem]{Definition}

\newtheorem{proposition}[theorem]{Proposition}

\newenvironment{proof}[1][Proof]{\noindent\textbf{#1.} }{\ \rule{0.5em}{0.5em}}
\usepackage{pifont}
\usepackage{geometry}
\usepackage{setspace}
\titlespacing*{\section}{0.00cm}{1.00cm}{0.50cm}
\titlespacing*{\subsection}{0.00cm}{0.50cm}{0.30cm}
\titlespacing*{\subsubsection}{0.50cm}{0.50cm}{0.30cm}
\makeatletter
\renewcommand*{\@seccntformat}[1]{\csname the#1\endcsname.\hspace{0.25cm}}
\makeatother

\setlist{nolistsep}
\begin{document}

\title{\textbf{The Social Welfare Implications\\of the Zenga Index}}
\author{
Francesca Greselin\thanks{%
Department of Statistics and Quantitative Methods, Universit\`{a} degli Studi di Milano-Bicocca. ORCID: 0000-0003-2929-1748. Phone: 0039 02 64 483 123. Email: francesca.greselin@unimib.it. Address: Via Bicocca degli Arcimboldi 8, 20126, Milano (IT).},
Simone Pellegrino\thanks{%
Corresponding Author. Department of Economics and Statistics -- ESOMAS, Universit\`{a} degli Studi di Torino. ORCID: 0000-0001-8372-1054. Phone: 0039 011 670 6060. Email: simone.pellegrino@unito.it. Address: Corso Unione Sovietica 218bis, 10134, Torino (IT).},
Achille Vernizzi\thanks{%
Department of Economics, Management and Quantitative Methods, Universit\`{a} degli Studi di Milano. ORCID: 0000-0002-1641-5003. Phone: 0039 02 50 321 460. Email: achille.vernizzi@unimi.it. Address: Via Conservatorio 7, 20122, Milano (IT).}\\
}

\date{June 18, 2020\\ \textbf{Preliminary draft -- Do not cite}}

\maketitle

\vspace{1cm}

\begin{abstract}
\noindent We introduce the social welfare implications of the \cite{zenga2007} index, a recently proposed index of inequality. Our proposal is derived by following the seminal book by \cite{son2013equity} and the recent working paper by \cite{KakwaniSon2019}. We compare the Zenga based approach with the classical one, based on the Lorenz curve and the Gini coefficient, as well as the Bonferroni index. We show that the social welfare specification based on the Zenga uniformity curve presents some peculiarities that distinguish it from the other considered indexes. The social welfare specification  presented here provides a deeper understanding of how the Zenga index evaluates the inequality in a distribution. 
\end{abstract}

\vspace{1.5cm}

{\bf JEL-Codes:} H23, H24.

{\bf Keywords\hspace{0.19cm}:} Zenga Index, Gini Index, Bonferroni index, Social Welfare, Bottom-to-top Ratios.

\newpage 
\section{Introduction}\label{intro}

Measuring and comparing social welfare levels attached to different income distributions is one of the most important issues in the study of income inequality. Summary statistic measures aseptically depict the inequality observed in a given distribution of income; but, from an economic perspective, the fundamental tool is given by the possibility to rank different distributions according to the social welfare level.

The seminal paper by \cite{Atkinson1970} illustrated under which conditions it is possible to rank distributions according to their welfare level, by focusing on the set of increasing and concave social welfare functions derived by additively separable and symmetric utility functions of individual incomes \citep{dalton1920}. Starting from Atkinson's pioneering paper and the related literature (see, among others, \cite{dasgupta_1973}, \cite{Rothschild1973}, \cite{Blackorby1977}, \cite{sen_1979}, \cite{usher_1980}, \cite{willig_1981} and \cite{Shorrocks1983}), the attention on the Lorenz curves and the generalized Lorenz curves follows.

The recent book by \cite{son2013equity} focuses the attention to the link between inequality measures and social welfare levels (Ch. 2): In particular, the author re-examines the \cite{Gini1914} coefficient social welfare implications on the light of the results due to \cite{Sen1974}, and then extends the results to the \cite{bonferroni_1930} coefficient.\footnote{\cite{KakwaniSon2019} enlarge the analysis by considering taxes.} Following this approach, the social welfare level imputable to an income distribution depends on the income levels and the attached weights. For what concerns the \cite{Gini1914} and the \cite{bonferroni_1930} coefficients, these weights depend solely on the ranks of income units.

The original contribution of the present paper is to apply and extend these results to the \cite{zenga2007} index and curve, a recent methodology to plot and measure inequality. Basically, the Zenga approach is based on comparing the mean income of the poorest income earners to the mean income of the remaining richest part of the population. The interest received by this fresh rethinking about inequality is documented by a still-growing literature. For example, we find research on properties of the index (see \cite{polisicchio2008, polisicchio2009,maffenini2014,arcagni_porro2014}), inferential theorems and applications (see \cite{greselin_pasquazzi2009, greselin_puri_zitikis2009, greselin_pasquazzi_zitikis2010, antal2011, langel_tille2012, greselin_pasquazzi_zitikis2013, greselin_pasquazzi_zitikis2014}), the decomposition of the index by population subgroups (\cite{radaelli2008, radaelli2010}), the decomposition by income sources (see \cite{ zenga2012, pasquazzi2018components}), a longitudinal decomposition  (\cite{mussini_zenga2013}), and application on real data (see, among others, \cite{arcagni_zenga2013, alina2019}).

Our findings show that, according to this new inequality index, the attached social welfare level depends on the inequality within the income distribution, and the weights attached to each income unit, as well.
Differently from the standard approach and the standard coefficients, in the Zenga index the role of the weight attached to each income unit is twofold and can be decomposed in two effects: the first effect solely considers the ranks of the income units, whilst the second depends on the ratio between the overall mean income and the average income of richest part of the income units. The novelty arises from the peculiar fact that, while the denominator for the Gini and the Bonferroni indexes is a constant value, i.e., the mean income of the population, in the Zenga index the denominator is a function of the rank of the income units. We then give evidence about the social welfare specification based on the Zenga uniformity curve. It presents some distinguishing features from the other considered indexes. In this context, the social welfare specification helps to better understand how the Zenga index evaluates the inequality in a distribution.\\

The remainder of this paper is organized as follows. Section \ref{basic_notation} presents the basic notations for the Gini and the Bonferroni as well as the Zenga coefficients. Section \ref{dominance} recalls the notion of Lorenz dominance. Then Section \ref{classical_SWF} addresses the issue of the social welfare evaluation according to the standard Gini coefficient, the generalized one, and the Bonferroni index. Section \ref{SWF_ZENGA} gets to the heart of the topic by presenting the peculiarities of the social welfare function according to the Zenga inequality index. 
Section \ref{conclusioni} offers some concluding remarks.

\section{Basic notation}\label{basic_notation}

Given a random variable $X\geq0$ with non negatively supported $cdf$ $F(X)$ for $X\geq0$, representing gross or net incomes as well as taxes, we denote by $F^{-1}(p)=inf\lbrace z:F(X)\geq p \rbrace$ the corresponding population quantile function for $0<p<1$.

The notion of the \cite{lorenz1905} curve was introduced to plot the cumulative share of $X$, denoted by $L_F(p)$, versus the cumulative share of the population $p$.\footnote{%
In the ideal case of perfect equality (that is, a society in which all people have the same income) the share of incomes equals the share of the population, so that $L_F(p)=p$, for all $0<p<1$. In this case the Lorenz curve is the diagonal line from $(0,0)$ to $(1,1)$. On the other hand, the lower the share of income $L_F(p)$ held by the share of income earners $p$, the higher the inequality. In the ideal case of perfect inequality (that is, a society in which all people but one have an income of nil) the share of incomes equals zero for $0 \leq p<1$, so that $L_F(p)=0$, and only for $p=1$ we have $L_F(1)=1$.} It is given by $(p,L_F(p))$, where
\begin{equation}
L_F(p)=\frac{\int_{0}^{p}F^{-1}(s)ds}{\int_{0}^{1}F^{-1}(s)ds}=\frac{1}{\mu_F}\int_{0}^{p}F^{-1}(s)ds	\label{Lorenz}
\end{equation}
\noindent and $\mu_F=E(X)$ denotes the mean value or the expectation of the random variable $X$.

It seems very natural to express the degree of inequality through the deviation of the actual Lorenz curve from the diagonal line. The \cite{Gini1914} index is precisely given by twice the area between the equality line and the Lorenz curve
\begin{equation*}
G_F=2\int_{0}^{1}(p-L_F(p))dp.
\end{equation*}

From another point of view, \cite{bonferroni_1930} proposed to express the degree of inequality through 
\begin{equation*}
B_F=\int_{0}^{1}(1-L_F(p)/p)dp.
\end{equation*}

More recently, \cite{zenga2007} introduced his proposal of measuring the inequality in the population by the following curve:
\begin{equation*}
Z_F(p)=\frac{L_F(p)}{(1-L_F(p))}\frac{(1-p)}{p},
\end{equation*}
and index
\begin{equation*}
Z_F=\int_{0}^{1}\frac{L_F(p)}{(1-L_F(p))}\frac{(1-p)}{p} \,dp.
\end{equation*}
\section{Dominance relationships}\label{dominance}
We begin recalling the notion of Lorenz dominance \citep{whitmore_1970}. We will see that two distributions are ordered with respect to the Lorenz ordering if and only if they are ordered with respect to the ordering based on the $Z_F(p)$ curve, too. Let $F_X$ and $F_Y$ be two distribution functions related to two continuous non-negative random variables $X$ and $Y$, both with finite and positive expected
value. We need here to introduce some notions of dominance:
 \begin{definition}
We say that  $X$ dominates $Y$ under the Lorenz ordering, denoting it by $F_X \geq_L F_Y $ if and only if $ L_{F_X} (p) \leq L_{F_Y} (p) \quad \forall p \in (0,1)$.
 \end{definition}
 \begin{definition}

We say that  $X$ dominates $Y$ under the Zenga ordering, denoting it by $F_X \geq_{Z} F_Y $ if and only if $ Z_{F_X} (p) \geq Z_{F_Y} (p) \quad \forall p \in (0,1)$. 
 \end{definition}
\cite{polisicchio2009} have shown that
\begin{eqnarray*}
 F_X \leq_L F_Y \iff L_X(p)&=&1- \frac{1-p}{1-pZ_X(p)} \geq 1- \frac{1-p}{1-pZ_Y(p)}= L_Y(p) \quad \forall p \in (0,1) \iff \\ \iff Z_X(p)&=&\frac{1}{p}-\frac{1-p}{p [1-L_X(p)] }\leq \frac{1}{p}-\frac{1-p}{p [1-L_Y(p)] }=Z_Y(p) \quad \forall p \in (0,1) \iff \\
 \iff F_X &\geq_{Z}& F_Y .
 \end{eqnarray*}

Therefore, introducing the Uniformity curve $U_X(p)=1-Z_X(p)$ that assesses the deviation, at the percent point $p$, from the case of equality, we obtain that the Lorenz curve of $X$ is always lower than the Lorenz of $Y$ if and only if the Uniformity curve of $X$ is always lower than the Uniformity curve of $Y$.

\section{Social Welfare functions for Gini, generalized Gini, and Bonferroni Indexes}\label{classical_SWF}
The welfare implications of the Gini index have been widely discussed in the literature by many authors, among them we may cite \cite{Atkinson1970}, \cite{newbery_1970}, \cite{kats1972}, \cite{Sheshinski1972}, \cite{dasgupta_1973}, \cite{Rothschild1973}, \cite{Sen1974}, \cite{Chipman1974}, \cite{sen_1973}. 

We  will recall briefly in this Section  the Gini and Bonferroni social welfare functions. To this aim, we start from the general form of the welfare function, given by
\begin{equation}
W=\int_0^{+\infty} x \,\, \nu(F(x)) \,\, f(x) \, \, dx
 \label{soc_wel_fun}.
\end{equation}
where $f(x)$ is the density function of $X$, $\nu(F(x)) $ is the weight attached to the income level $x$ such that $\nu'(F(x)) <0$, implying that weights must decrease monotonically with $F(x)$; in other words, greater weights are given to poorer persons than richer ones, and the total weight adds up to 1 for the entire population:
\begin{equation*}
\int_0^{+\infty} \,\, \nu(F(x)) \,\, f(x) \, \, dx =1.
\end{equation*}

The social welfare function captures the idea of relative deprivation, by assigning a weight to income $x$ to depend on the ranking of all individuals in society. The lower a person is on a welfare scale, the higher this person's sense of deprivation with respect to others in society. We call this a class of \textit{rank order social welfare functions}. 

\cite{Shorrocks1983} and \cite{kakwani1984APK} extended the usual ranking of distributions based on the Lorenz order, and proved that if the generalized Lorenz curve $\mu L(p)$ of the distribution $X_1$ is higher than the 
generalized Lorenz curve of the distribution $X_2$ at all points, then
for a wide range of social welfare functions (all symmetric and quasi-concave s.w.f.), the social welfare implied by  $X_1$
will always be higher than the social welfare implied by  $X_2$. Therefore, the area under twice the generalized Lorenz can be used as a measure of social welfare
\begin{equation*}
W_G=2\int_{0}^{1} \mu L_F(p) \,dp = \mu [1-G],
\end{equation*}
yielding the social welfare function implied by the Gini index.
The last equality is derived as follows:

\begin{eqnarray*}
W_G&=&2\int_{0}^{1} \mu L_F(p) \,dp \\
&=&\mu \left[ 1-\left( 1-2 \int_{0}^{1} L_F(p) \,dp \right) \right] =\mu[1-G]. \\
\end{eqnarray*}

We would like now to express $W_G$ like a welfare function, by making explicit the weights $\nu_G(F(x))$ for the welfare function implicit within the Gini approach:
\begin{eqnarray*}
W_G&=&2 \int_{0}^{1} \mu L_F(p) \,dp \\
&=& 2 \int_0^1 \left( \int_0^p F^{-1}(s) \, ds \right) \, dp \\
&=& 2 \left[ \left( \int_0^p F^{-1}(s) \, ds \right) p \bigg{|}_0^1 - \int_0^1 F^{-1}(p) \,\, p \,\, dp \right] \\
&=& 2 \left[ \int_0^1 F^{-1}(p) \,\, dp  - \int_0^1 F^{-1}(p) \,\, p\,\, dp \right] \\
&=& 2  \int_0^1 F^{-1}(p) \, (1-p) \,\, dp.
\end{eqnarray*}

We have therefore obtained that $\nu_G(p)=2(1-p)$ is the weight attached to the income level $x_p=F^{-1}(p)$, following the Gini index. The  change of variable from $p=F(x)$ to $x$ provides the way to read the Gini welfare function as a special case of \eqref{soc_wel_fun}, as follows
\begin{eqnarray*}
W_G&=& \int_{0}^{\infty} x \,\, 2(1-F(x))\,\, f(x)dx.
\end{eqnarray*}

\cite{kakwanilibro1980} and \cite{Donaldson1980} as well as \cite{Donaldson1983} developed versions of the extended Gini that depend on social welfare functions.
The Social Welfare function implicit in the generalized Gini index $G_k$ has been derived by \cite{KakwaniSon2019}, analogously, leading to
\begin{equation*}
W_{G_k}= (k+1) \int_{0}^{1} x [1-F(x)]^k f(x) \,\,dx = \mu[1-G_k].
\end{equation*}

We see that $\nu_{G_k}(p) = (k+1)(1-p)^k$ is the weight attached to the income level 
$p = F^{-1}(p)$, following the generalized Gini index.

For the Bonferroni index, \cite{son2013equity} derived a measure of social welfare that is equal to
the area under the generalized Bonferroni curve
\begin{equation*}
W_B= \mu \int_{0}^{1} \frac{L_F(p)}{p} \,dp = \mu[1-B].
\end{equation*}

We would like to express $W_B$ like a welfare function, by making explicit the weights $\nu_B(F(x))$ for the welfare function implicit within the Bonferroni approach:
\begin{eqnarray*}
W_B&=&\mu \int_{0}^{1} \frac{L_F(p)}{p} \,dp\\
&=& \int_0^1 \left( \int_0^p F^{-1}(s) \, ds \right)\,\, \frac{1}{p} \,\, dp \\
&=& \left( \int_0^p F^{-1}(s) \, ds \right) \ln p \bigg{|}_0^1 - \int_0^1 F^{-1}(p) \,\, \ln p \,\, dp \\
&=& - \int_0^1 F^{-1}(p) \,\, \ln p\,\, dp. \\
\end{eqnarray*}

We have shown that $\nu_B(p)= -\ln p$ is the weight attached to the income level $x_p=F^{-1}(p)$, following the Bonferroni approach \citep{KakwaniSon2019}.

\section{Social welfare function for the Zenga index}\label{SWF_ZENGA}

Finally, we are ready to introduce the social welfare function for the Zenga approach. Due to the equivalence of the orderings given by the Lorenz and the Zenga index, recalled in the previous section, we follow the same steps we did in the previous section,
\begin{equation*}
W_Z=\mu \int_{0}^{1}\frac{L_F(p)}{(1-L_F(p))}\frac{(1-p)}{p} \,dp = \mu[1-Z].
\end{equation*}

As we did with the previous approaches, we would like to express the last integral in terms of a weight function $\nu_Z(F(x))$ to be attached to the income level $x_p=F^{-1}(p)$. To this aim, we have the following result:
\begin{eqnarray}
W_Z=\mu \frac{L_F(p)}{(1-L_F(p))}( \ln p +1-p ) \bigg{|}_0^1 + \int_{0}^{1}\frac{F^{-1}(p)}{(1-L_F(p))^2}(- \ln p + p - 1 ) \,dp 
\label{sarebbebello}
\end{eqnarray}
that holds true by considering integration by parts, $\int_a^b f \,\,g' \,\,dp = f \,\,g \,\,\big{|}_0^1 - \int_a^b f' \,g \,\, dp$, with 
\begin{eqnarray*}
g(p)&=&\ln p + 1-p \quad \textrm{and} \quad g'(p)=\frac{1}{p}-1; 
\end{eqnarray*}
and
\begin{eqnarray*}
 f(p)&=&\frac{L_F(p)}{(1-L_F(p))}\quad \textrm{with} \quad f'(p)=\frac{F^{-1}(p)}{\mu (1-L_F(p))^2}.
\end{eqnarray*}
\begin{proposition}

The first term in the right hand side of \eqref{sarebbebello} is equal to zero.
\label{proposition1}
\end{proposition}
\begin{proof}
 We compute
\begin{equation*}
\lim_{p \rightarrow 0} \frac{L_F(p)}{(1-L_F(p))}( \ln p +1-p ) =\lim_{p \rightarrow 0} \frac{\frac{L_F(p)}{1-L_F(p)}}{p} \,\,\, \frac{( \ln p +1-p )}{\frac{1}{p}}= \left[ \frac {0}{0} \right] \,\,\,\left[ \frac {-\infty}{\infty} \right]
\end{equation*}
and the indeterminate form is solved by using De l'Hopital rule on each term:
\begin{equation*}
\lim_{p \rightarrow 0} \frac{\frac{L_F(p)}{1-L_F(p)}}{p} =\lim_{p \rightarrow 0} \frac{L_F'(p)(1-L_F(p))+L_F(p)L_F'(p)}{(1-L_F(p))^2}=\lim_{p \rightarrow 0} \frac{x_p}{\mu}
\end{equation*}
and
\begin{equation*}
\lim_{p \rightarrow 0} \frac{( \ln p +1-p )}{\frac{1}{p}}=\lim_{p \rightarrow 0} \frac{1-\frac{1}{p}}{\frac{1}{p^2}} =0.
\end{equation*}

 We have still to evaluate
\begin{equation*}
 \lim_{p \rightarrow 1} \frac{L_F(p)( \ln p +1-p ) }{(1-L_F(p))} = \left[ \frac {0}{0} \right].
\end{equation*}
Also in this case we apply De l'Hopital rule, leading to
\begin{equation*}  
 \lim_{p \rightarrow 1} \frac{\frac{x_p}{\mu}( \ln p +1-p )+L_F(p)\left( \frac{1}{p} -1\right)}{-\frac{x_p}{\mu}}=\\
 \lim_{p \rightarrow 1} -( \ln p +1-p )- \lim_{p \rightarrow 1} \frac {L_F(p) \mu}{ x_p}\left( \frac{1}{p} -1\right)=0.
 \end{equation*}
 \end{proof}
 
Coming back to \eqref{sarebbebello}, and using Proposition \ref{proposition1}, we conclude  that the weight $\nu_Z(p)$ to be given to each observation $x_p$ assumed within the Zenga approach is given by
\begin{eqnarray*}
\nu_Z(p)=\frac{(- \ln p +p-1 )}{(1-L_F(p))^2},
\label{sarebbebello1}
\end{eqnarray*}   
where $p=F(x)$. In alternative, we may want to express the weight $\nu_Z(p)$ as a function of $x$, in the integral \eqref{sarebbebello}, by making the change of variable $x=F^{-1}(p)$
\begin{equation}
W_Z=\mu \int_{0}^{+\infty}\,\,x\,\, \frac{(- \ln F(x) + F(x)-1)}{(1-L_F(F(x)))^2} \,\,f(x)\,\,dx.
\label{sarebbebello2}
\end{equation}

After deriving $\nu_Z(p)$, we decompose the weight function for the Zenga Social Welfare function, into two terms, to isolate the contribution given by the ranks, and the contribution given by the interplay between the inequality and the ranks. We define

\begin{eqnarray*}
\nu^*_Z(p)&=&\frac{(- \ln p +p-1 )}{(1-p)^2}, \quad \quad \textrm{and}\\
\beta_Z(p)&=&\frac{(1-p )^2}{(1-L_F(p))^2},
\label{decomposition}
\end{eqnarray*}
so that we obtain $\nu_Z(p)=\nu^*_Z(p)\,\,\beta_Z(p)$.

We will show that the term $\nu^*_Z(p)$ is a \textit{non negative, concave upward and strictly decreasing} function of the rank $p$, like the Social Welfare functions 
\begin{equation}
\nu_G (p)=2\,(1-p) \quad \quad \nu_{G_k} (p)=(k+1)\,(1-p)^k \quad \quad \textrm{and} \quad \quad \nu_B(p)=-\ln (p)
\end{equation}
that are implicit in the Gini index $G$, the generalized Gini $G_k$, and the Bonferroni index $B$, respectively (see \cite{KakwaniSon2019}). We will show that the total weight $\nu^*(p)$ adds up to 1. Therefore, the function $\nu^*(p)$ incorporate a society's distributional judgement, where the poorest individual receives the maximum weight and the richest individual gets the minimum weight.

The second term
\begin{equation*}
\beta_Z(p)=\frac{(1-p )^2}{(1-L_F(p))^2}=\Bigg(\frac{\mu}{\mu^+_{F}(p)}\Bigg)^2
\end{equation*}
depends, instead, on both $p$ and $L(p)$. It is straighforward to see, from its rightmost expression, that $\beta_Z(p)$ is a decreasing function of $p$, with $\beta_Z(0)=+ 1$ and $ \lim_{p \rightarrow 1} \beta_Z(p)=0$. 
The greater the ratio $\frac{\mu^+_{F}(p)}{\mu}$ the greater the penalization given to the quantile $x_p$ by $\beta_Z(p)$. 
Comparing to what we have recalled about the Gini and Bonferroni indexes, in the social welfare evaluation based on the Zenga Uniformity curve, beyond the weight function based on the ranks $p$, we also have $\beta_Z(p)$.

This difference arises from the fact that the denominator for the Gini and the Bonferroni indexes is a constant value, namely the mean (income), while in the present approach, the denominator is a function of $p$.

Now, we want to make a list of the crucial points that we want to address, to show that our methodology yields a Social welfare function complying with \cite{Atkinson1970} requirements. We are interested in proving the following required properties:
\begin{itemize}
\item The weight function $\nu^*_Z(F(x))$ is a decreasing function, giving lesser weight to richer individuals in society. Moreover, by studying its second derivative, we are interested in studying its behavior. We will study, analogously, the behavior of the curve $\nu_Z(p)$, while for $\beta_Z(p)$, we have already derived the needed information.
\item We want to obtain the limits for the weight function $\nu_Z(F(x))$ over the interval $(0,1)$.
\item By integrating the weight function $\nu^*_Z(F(x))$ over $(0,1)$, we obtain 1, that is the total weight add to unity. 
\end{itemize}

 \subsection{Properties of the weight function}
 In this section we would like to derive the behaviour of the functions $\nu^*_Z(F(x))$, 
 and $\nu_Z(p)$.
 We begin by obtaining the limits of the two 
 functions on the extremes of the interval $[0,1]$, then we study the first derivative and second derivatives.
 
 \subsubsection{Limits of $\nu_Z(p)$ and $\nu^*_Z(p)$} 
We are interested in deriving the limits of the weight function $\nu_Z(p)$ for $p \rightarrow 0$ and $p \rightarrow 1.$

The first is immediate
\begin{eqnarray*}
\lim_{p \rightarrow 0} \frac{( -\ln p +p-1 )}{(1-L_F(p))^2}= + \infty ,
\end{eqnarray*}
while the second is more tricky:
\begin{eqnarray*}
\lim_{p \rightarrow 1} \frac{(- \ln p +p-1 )}{(1-L_F(p))^2}= \left[ \frac{0}{0} \right]
\end{eqnarray*}

We come back to the useful decomposition of $\nu_Z(p)$ into the product of $\nu^*_Z(p)$ and $\beta_Z(p)$, and we manage each term separately, 
\begin{eqnarray*}
\lim_{p \rightarrow 1}\nu^*_Z(p)=\lim_{p \rightarrow 1} \frac{( \ln p +1-p )}{(1-p)^2}= \left[ \frac{0}{0} \right] \quad \textrm{and} \quad \lim_{p \rightarrow 1} \beta_Z(p)=\lim_{p \rightarrow 1} \left(\frac{\mu}{\mu_F^+(p)}\right)^2=0
\end{eqnarray*}
where the last equality yields from considering the distribution $F$ having infinite support on the positive real line.
After using De l'Hopital on the first limit on $\nu^*_Z(p)$, 
\begin{eqnarray*}
\lim_{p \rightarrow 1} \frac{1-\frac{1}{p}}{2(p-1)}=\left[ \frac{0}{0} \right]
\end{eqnarray*}
we still get to an indeterminate form, so we apply De l'Hopital again:
\begin{eqnarray*}
\lim_{p \rightarrow 1} \frac{1}{2p^2}=1/2.
\end{eqnarray*}

To summarize, we have shown\footnote{In case of dealing with a finite sample, naturally, there is a maximum value of the income, say $max \{x_i ; i=1,\ldots,n\} =x_{max}$ and for the empirical version $\beta_{Z_n}(p)$ of $\beta_{Z}(p)$ it holds $ \lim_{p \rightarrow 1} \beta_{Z_n}(p)= (\overline{X}/x_{max})^2$, where $\overline{X}$ is the empirical mean.} that  
\begin{eqnarray*}
\lim_{p \rightarrow 0} \nu_Z(p)&=&+ \infty, \quad \textrm{and} \quad \lim_{p \rightarrow 1} \nu_Z(p)=0, \\
\lim_{p \rightarrow 0} \nu^*_Z(p)&=&+ \infty, \quad \textrm{and} \quad \lim_{p \rightarrow 1} \nu^*_Z(p)=1/2. 
\end{eqnarray*}

\subsubsection{The weight $\nu^*_Z(p)$ is a decreasing function of $p$}
 \begin{eqnarray*}
{\nu^*}'_Z(p)&=&\frac{d}{dp}\frac{(- \ln p +p-1 )}{(1-p)^2} \\
&=&\frac{\left(-\frac{1}{p}+1\right)(1-p)^2- (- \ln p +p-1 )2(1-p)(-1)}{(1-p)^4}\\
&=&\frac{-\frac{(1-p)^2}{p}+2(- \ln p +p-1 )}{(1-p)^3}.
\end{eqnarray*}

Now we have to study the sign of the numerator of the last expression, that can be rewritten as follows
\begin{eqnarray*}
-2 \ln p -\frac{1-p^2}{p}.
\end{eqnarray*}

Let us set 
\begin{eqnarray*}
h(p)=-2 \ln p, \quad \textrm{and} \quad g(p)=\frac{1-p^2}{p},
\end{eqnarray*}
and observe that
\begin{eqnarray*}
\lim_{p \rightarrow 0} h(p)= \lim_{p \rightarrow 0} g(p)=+\infty, \quad \textrm{and} \quad \lim_{p \rightarrow 1} h(p)= \lim_{p \rightarrow 1} g(p)=0. 
\end{eqnarray*}

By studying the derivatives
 \begin{eqnarray*} 
h'(p)=-2/p, \quad \textrm{and} \quad g'(p)=-(1+p^2)/p^2
\end{eqnarray*}
 which are both negative $\forall p \in (0,1]$, and such that $|h'(p)| \leq |g'(p)|$, we conclude that $h(p) \leq g(p)$.
 
 This means that 
 \begin{eqnarray*} 
 {\nu^*}'_Z(p)\leq 0 \quad \quad \forall p \in (0,1].
 \end{eqnarray*}
  
\subsubsection{The weight $\nu^*_Z(p)$ sum up to 1}
 To show this property, we start from
 \begin{eqnarray*}
 \frac{d}{dp}\frac{(- \ln p +p-1 )}{(1-p)} = \frac{(- \ln p +p-1 )}{(1-p)^2} -\frac{1}{p}
\end{eqnarray*}
that helps us in solving 
 \begin{eqnarray*}
 \int_0^1 \frac{(- \ln p +p-1 )}{(1-p)^2} dp&=& \frac{(- \ln p +p-1 )}{(1-p)} \bigg|_0^1+\int_0^1 \frac{1}{p} dp	\\
 &=& \left[\frac{(- \ln p +p-1 )}{(1-p)} +\ln p \right.\bigg|_0^1	\\
 &=&\left[-1-\frac{\ln p}{(1-p)/p}\right.\bigg|_0^1.
\end{eqnarray*}

The latter expression leads to two indeterminate forms, of type $\left[\frac{0}{0}\right]$ and $\left[\frac{-\infty}{+\infty}\right]$ respectively. By De l'Hopital we have
 \begin{eqnarray*}
-\frac{\ln p}{(1-p)/p}\bigg|_0^1=\lim_{p\rightarrow 1} \,\frac{1/p}{1/p^2}\,\,-\lim_{p\rightarrow 0} \,\frac{1/p}{1/p^2}\,=1
\end{eqnarray*}
from which we get 
 $\int_0^1 \nu^*_Z(p) dp=1$.
 

\subsubsection{The weight $\nu_Z(p)$ is a decreasing function of $p$}

 From the definition, $\nu_Z(p)=\nu^*_Z(p)\,\,\beta_Z(p)$, and due to $\nu'_Z(p)={\nu^*}'_Z(p)\,\,\beta_Z(p)+\nu^*_Z(p)\,\,\beta'_Z(p)$ and the previous results, we get the thesis.

\subsubsection{The weight $\nu^*_Z(p)$ is an upward concave function of $p$}

\begin{eqnarray*}
{\nu^*}''_Z(p)&=&\frac{d^2}{dp^2}\left[\frac{(- \ln p +p-1 )}{(1-p)^2} \right]\\
&=&\frac{d}{dp}\left[\frac{-\frac{1}{p}+p-2 \ln p }{(1-p)^3}\right]\\
&=&\frac{\left(\frac{1}{p^2}+1- \frac{2}{p}\right)(1-p)^3+3(1-p)^2\left(- \frac{1}{p}+p-2\ln p\right) }{(1-p)^6}\\
&=&\frac{\frac{1}{p^2}- \frac{6}{p}+2p+3-6\ln p }{(1-p)^4}.
\end{eqnarray*}

To show the upward concavity of $\nu^*_Z(p)$ we will show that the numerator of the last ratio is non negative, that is
\begin{eqnarray*}
\frac{2p^3+3p^2-6p+1}{p^2} \geq 6 \ln p.
\end{eqnarray*}

We decompose $h(p)=\frac{1}{p^2}\left(2p^3+3p^2-6p+1\right)=\frac{2}{p^2}(p-1)\left(p-\frac{-5+\sqrt{33}}{4}\right)\left(p+\frac{-5+\sqrt{33}}{4}\right)$
and see that $h(p) \geq 0\geq 6\ln p$ for $0\leq p \leq \frac{-5+\sqrt{33}}{4}\approx0.186$.

For $\frac{-5+\sqrt{33}}{4}\leq p \leq 1$, we have that $h(p)-6\ln p$ is a monotonically decreasing function, and gets values in $[0,-6\ln \frac{-5+\sqrt{33}}{4}]$. Indeed:
\begin{eqnarray*}
\frac{d}{dp}\left[\frac{2p^3+3p^2-6p+1}{p^2} - 6 \ln p\right]&=&\\
&=&\frac{\left(6p-6+6p^2\right)p^2-2\left(1+3p^2-6p+2p^3\right)p}{p^4}-\frac{6}{p}\\
&=&\frac{(-2+6p+2p^3-6p^2)}{p^3}\\
&=&\frac{2(p^3-1)}{p^3}\leq 0.
\end{eqnarray*}

\section{Concluding Remarks}\label{conclusioni}
This paper is a first contribution to derive, analyze, and understand, the social welfare implications of the Zenga equality curve and index. Following the existing literature, and the recent publications by \cite{son2013equity} and \cite{KakwaniSon2019}, we first introduce the social welfare function according to the Zenga methodology and then propose how to interpret the weight function to be attached to each income level in order to compute the social welfare level. In particular, differently from the Gini and Bonferroni coefficients, the weights of the Zenga index $\nu_Z(p)$ can be split into two different contributions: the first one, which we call $\nu^*_Z(p)$, that depends only on the income ordering and has a non-negative, concave upward and strictly decreasing behavior, and the second one, $\beta_Z(p)$, also decreasing with respect to the income ordering, which is related to the inequality of the income distribution under consideration. We discuss the properties of these elements to address the peculiarities of the new approach.



\bibliographystyle{plainnat}


\newpage

\appendix
\renewcommand\thefigure{\thesection.\arabic{figure}}  
\renewcommand\thefigure{\thesection.\arabic{table}}

\end{document}